\documentclass[11pt]{article}

\usepackage{amssymb} 
\usepackage{amsmath}     
\usepackage{amsthm}

 \usepackage{eufrak}

 \usepackage{graphicx}               
 \usepackage{psfrag}

\usepackage{a4wide}

\usepackage[ruled,vlined]{algorithm2e}
\SetKw{KwNot}{not}
\SetKw{KwExitFor}{exit for-loop}

\newtheorem{thm}{Theorem}
\newtheorem{lem}{Lemma}

\newtheorem{cor}{Corollary}

\newenvironment{keyword}{\par{\noindent\bf Keywords:}}

\begin{document}

\title{On the approximability of 
robust spanning tree problems}

\author{Adam Kasperski\\
   {\small \textit{Institute of Industrial}}
  {\small \textit{Engineering and Management,}}
  {\small \textit{Wroc{\l}aw University of Technology,}}\\
  {\small \textit{Wybrze{\.z}e Wyspia{\'n}skiego 27,}}
  {\small \textit{50-370 Wroc{\l}aw, Poland,}}
  {\small \textit{adam.kasperski@pwr.wroc.pl}}
  \and
  Pawe{\l} Zieli{\'n}ski\\
    {\small \textit{Institute of Mathematics}}
  {\small \textit{and Computer Science}}
  {\small \textit{Wroc{\l}aw University of Technology,}}\\
  {\small \textit{Wybrze{\.z}e Wyspia{\'n}skiego 27,}}
  {\small \textit{50-370 Wroc{\l}aw, Poland,}}
  {\small \textit{pawel.zielinski@pwr.wroc.pl}}}

\date{}
\maketitle

\begin{abstract}
In this paper the minimum spanning tree problem with uncertain edge costs is discussed. In order to model the uncertainty a discrete scenario set is specified and  a robust framework is adopted to choose a solution.
The min-max, min-max regret and 2-stage min-max versions of the problem are discussed. The complexity and approximability of all these problems are explored.
It is proved that the min-max and min-max regret versions with nonnegative edge costs  are hard to 
approximate  within $O(\log^{1-\epsilon} n)$ 
for any $\epsilon>0$ unless the problems in NP have quasi-polynomial time
algorithms. Similarly, the 2-stage min-max problem  cannot be approximated within 
  $O(\log n)$  unless the problems in NP have quasi-polynomial time
  algorithms. In this paper randomized LP-based approximation algorithms with performance ratio of $O(\log^2 n)$ for min-max and 2-stage min-max problems are also proposed.

%

\end{abstract}


\section{Introduction}

The usual
assumption in 
combinatorial optimization
is that all input parameters are
precisely known. However, in real life this is rarely the case.
There are
two popular optimization settings of problems for hedging against 
uncertainty  of parameters: \emph{stochastic optimization setting}
and  \emph{robust optimization setting}.

In the stochastic optimization, the uncertainty is modeled 
by specifying probability distributions of the parameters and the goal is to optimize the expected value
of a solution built (see,~e.g.,~\cite{BF97,V07}).
One of the most popular models of the stochastic optimization is a
\emph{2-stage model}~\cite{BF97}. In the 2-stage approach the
precise values of the parameters are specified in the first stage,
while the values of these parameters in the second stage are uncertain and
are specified by probability distributions.
The goal is to choose a part of a solution in
the first stage and complete it in the second stage
so that the expected value of the obtained solution is optimized.
Recently, there has been a growing interest
in 
combinatorial optimization
problems formulated in
the 2-stage stochastic 
framework~\cite{DRM05,EGMS10,FFK06,KMU08,RS07}.

In the robust optimization setting~\cite{KY97} the uncertainty is modeled by specifying a
set of all possible realizations of the parameters called \emph{scenarios}. No probability distribution in the scenario set is given. 
In the \emph{discrete
scenario case}, which is considered in this paper, we define a scenario set by explicitly listing all scenarios.
Then, in order to choose a solution,
two optimization criteria,
called the \emph{min-max} and  the \emph{min-max regret}, can be
adopted. 
Under the min-max criterion, we seek a solution that minimizes 
the largest cost over all scenarios. Under the min-max regret criterion we wish to find a solution which minimizes the largest deviation from optimum over all scenarios. A deeper discussion on both criteria can be found in~\cite{KY97}.
The minmax (regret) versions of some basic
combinatorial optimization problems with discrete structure of uncertainty 
 have been extensively studied  in the recent 
 literature~\cite{ABV05,ABV06,KZ09, MMS08}.
 Furthermore, both robust criteria can be easily  extended to the 2-stage framework.
 Such an extension has been recently done in~\cite{DGR06,KMU08}.

In this paper, we wish to investigate the min-max (regret) and min-max 2-stage versions of the classical \emph{minimum spanning tree} problem.  The classical deterministic problem is formally stated as follows.
We are given a connected graph~$G=(V,E)$
with edge costs $c_e$, $e\in E$. We seek a \emph{spanning tree} of $G$  of the minimal total cost. We use $\Phi$ to denote the set of all spanning trees of $G$. The classical deterministic minimum spanning tree is a well studied problem, for which several very efficient algorithms 
exist (see,~e.g.,~\cite{AMO93}).

In the robust framework, the edge costs are uncertain and the set of scenarios~$\Gamma$ is defined by explicitly listing all possible edge cost vectors.
So, $\Gamma=\{S_1,\dots,S_K\}$ is finite and contains exactly $K$ scenarios, where a \emph{scenario} is a cost realization $S=(c^S_e)_{e\in E}$. In this paper we consider the \emph{unbounded case}, where the number of scenarios is a part of the input. 
We will denote by 
$C^*(S)=\min_{T\in \Phi} \sum_{e\in T} c^{S}_e$ the cost of a 
minimum spanning tree under a fixed scenario $S\in\Gamma$. 
In the \textsc{Min-max Spanning Tree} 
problem, we seek a spanning tree that minimizes the largest
cost over all scenarios, that is
\begin{equation}
OPT_1=\min_{T\in \Phi}\max_{S\in\Gamma} \sum_{e\in T} c^{S}_e.
\label{pminmax}
\end{equation}
In the \textsc{Min-max Regret Spanning Tree},
 we wish to find a spanning tree that minimizes the maximal regret:
\begin{equation}
OPT_2=\min_{T\in \Phi} \max_{S\in \Gamma}
\left\{\sum_{e\in T} c^{S}_e-C^*(S)\right\}.
\label{pregret}
\end{equation}

The formulation~(\ref{pminmax}) is a single-stage decision one.
We can extend this formulation to a 2-stage case as follows.
We are given the first stage edge costs $c_e$, $e\in E$,
and 
in the second stage there are $K$ possible cost realizations (scenarios) listed in scenario set $\Gamma$.
The \textsc{2-stage Spanning Tree} 
problem consists in determining a subset of edges $E_1$
in the first stage and 
 a subset of edges $E^{S}_2$ that augments it
  to form a spanning tree $T^S=E_1\cup E^{S}_2\in \Phi$
 under scenario~$S$ in the second stage
  for each scenario~$S\in \Gamma$.
  The goal is minimize the maximum cost of the determined 
  subsets of edges $E_1$, $E^{S_1}_2,\ldots, E^{S_K}_2$: 
\begin{equation}
OPT_3=\min_{E_1,E^{S_1}_2,\ldots, E^{S_K}_2} 
\max_{S\in\Gamma}
\left\{ \sum_{e\in E_1} c_e  +
 \sum_{e\in E^{S}_2} c^{S}_e\;:\; T^{S}=E_1\cup E^{S}_2\in \Phi
  \right\}.
\label{p2stage}
\end{equation}

Let us now recall some known results on the problems under consideration.
In the bounded case (when the number of scenarios is bounded by a constant), the
\textsc{Min-max (Regret) Spanning Tree} problem
is NP-hard even if $\Gamma$ contains only~2 scenarios~\cite{KY97}
and
  admits an FPTAS~\cite{ABV06}, whose running time, however, grows exponentially with $K$.
In the  unbounded case,  
the \textsc{Min-max (Regret) Spanning Tree} problem
is strongly NP-hard~\cite{ABV05,KY97}
and not approximable 
within
$(2-\epsilon)$, for any $\epsilon>0$,
unless P=NP even for edge series-parallel graphs~\cite{KZ09}.
The \textsc{Min-max (Regret) Spanning Tree} problem
is approximable within $K$~\cite{ABV06}.
However,  up to now the existence of an 
approximation algorithm  with a constant performance ratio for the 
unbounded case has been an open question. 
To the
best of the authors' knowledge
the 2-stage version of the minimum spanning tree problem
seems to exist only 
in the stochastic setting~\cite{DRM05,EGMS10,FFK06}.
 Recently, the robust 2-stage framework has been employed 
 in~\cite{DGR06,KMU08}
 for some network design and matching problems.

\paragraph{Our results}
In this paper we prove that 
the \textsc{Min-max Spanning Tree}
and \textsc{Min-max Regret Spanning Tree} problems 
are hard to  approximate 
with a constant performance ratio (Theorem~\ref{cmmmst}
and Corollary~\ref{cmmemst}). Namely, they are
are not approximable within $O(\log^{1-\epsilon} n)$ 
for any $\epsilon>0$,
  where $n$ is the input size, unless NP 
  $\subseteq$ DTIME$(n^{\mathrm{poly} \log n})$.
  We thus give a
negative answer to the open question about
the existence of approximation algorithms 
with a constant performance ratio 
for these problems.
Moreover, if both positive and  negative edge costs
 are allowed, then the 
\textsc{Min-max Spanning Tree}
problem  is not at all approximable unless P=NP (Theorem~\ref{tmmna}). 
For the \textsc{2-stage Spanning Tree} problem, 
we show that  it is not approximable 
within   any constant, unless P=NP, and 
  within $(1-\epsilon)\ln n$
  for any $\epsilon>0$, unless
   NP$\subseteq$DTIME$(n^{\log \log n})$ 
   (Theorem~\ref{t2smsp}).  The above negative results encourage us to find
randomized approximation algorithms,
 which yield a $O(\log^2 n)$ approximation ratio
 for
 \textsc{Min-max Spanning Tree}
(Theorem~\ref{thm3}) and \textsc{2-Stage min-max Spanning Tree} (Theorem~\ref{tra2s}).

\section{Min-max (regret) spanning tree}

In this section, we study the
\textsc{Min-max  Spanning Tree} and \textsc{Min-max Regret Spanning Tree}
problems. 
We improve
 the results obtained in~\cite{ABV05, KZ09}, by showing that
 both problems are hard to approximate   within a ratio of
$O(\log^{1-\epsilon} n)$ for any $\epsilon>0$,
unless the problems in NP have
quasi-polynomial time algorithms.
We then provide an LP-based randomized  algorithm
with approximation ratio of $O(\log^2 n)$
for \textsc{Min-max  Spanning Tree}.

\subsection{Hardness of approximation}

We reduce a variant of the \textsc{Label Cover} problem
(see e.g., \cite{AC95,MMS08})
to \textsc{Min-max  Spanning Tree}.

\begin{description}
\item[\mdseries \scshape Label Cover:]

\emph{Input:}
A regular  bipartite graph $G=(V,W,E)$, $E\subseteq V\times W$;
an integer~$N$ that defines the set of labels, which are
in integers in $\{1,\ldots,N\}$;
for every edge $(v,w)\in E$ a partial map $\sigma_{v,w}:\{1,\dots,N\}\rightarrow \{1,\dots,N\}$.
A \emph{labeling} of the instance
$\mathcal{L}=(G,N,\{\sigma_{v,w}\}_{(v,w)\in E})$ is a function~$l$ assigning a nonempty 
set of labels to each vertex in $V\cup W$,
namely $l: V\cup W \rightarrow 2^N$. A labeling \emph{satisfies}
an edge $(v,w)\in E$ if
\[
\exists a\in l(v), \exists b\in l(w):\; \sigma_{v,w}(a)=b.
\]
A \emph{total labeling} is a labeling that satisfies all edges.
The value of   a total labeling~$l$ is $\max_{x\in V\cup W}|l(x)|$.

\noindent\emph{Output:}
A total labeling of the minimum value. This value is denoted by 
$val(\mathcal{L})$.

\end{description}

We now recall the following theorem~\cite{AC95, MMS08}:
\begin{thm}
There exists a constant~$\gamma>0$ so that
for any language $L\in NP$, any input $\mathbf{w}$ and $N>0$,
one can construct an instance~$\mathcal{L}$ of
\textsc{Label Cover},
with $|\mathbf{w}|^{O(\log N)}$ 
vertices and the label set of size~$N$,
so that:
\begin{eqnarray*}
\mathbf{w}\in L&\Rightarrow& val(\mathcal{L})=1,\\
\mathbf{w}\not\in L&\Rightarrow& val(\mathcal{L})\geq N^{\gamma}.
\end{eqnarray*}
Furthermore, $\mathcal{L}$ can be constructed in time
polynomial in its size.
\label{tlancover}
\end{thm}

We now state and prove  
the theorem,
which is 
essential in showing the hardness results
for the problems of interest. 
\begin{thm}
There exists a constant~$\gamma>0$ so that
for any language $L\in NP$, any input $\mathbf{w}$, any $N>0$
and any $g\leq N^{\gamma}$,
one can construct an instance~$\mathcal{T}$ of
\textsc{Min-max  Spanning Tree}
in time $O(|\mathbf{w}|^{O(g\log N)}N^{O(g)})$, so that:
\begin{eqnarray*}
\mathbf{w}\in L&\Rightarrow& OPT_1(\mathcal{T})\leq 1,\\
\mathbf{w}\not\in L&\Rightarrow& OPT_1(\mathcal{T}) \geq g.
\end{eqnarray*}
\label{taprminmax}
\end{thm}
\begin{proof}

Let $L$ be a language in $NP$ and let $\mathcal{L}=(G=(V,W,E),N,\{\sigma_{v,w}\}_{(v,w)\in E})$ be the instance of 
\textsc{Label Cover}  from Theorem~\ref{tlancover} constructed for $L$. Let us introduce some additional notations:
\begin{itemize}
		\item $\delta(x)$ is the set of edges of $G$ incident to vertex~$x\in V \cup W$,
		\item $N_{v,w}=\{(a,b)\in N \times N: \sigma_{v,w}(a)=b\}$. 
\end{itemize} 
We now transform $\mathcal{L}$ to an instance~$\mathcal{T}$ of \textsc{Min-max Spanning Tree}.  Let us fix $g\leq N^\gamma$, where $\gamma$ is the constant from Theorem~\ref{tlancover}. We first construct graph $G'$ in the following way. We replace every edge $(v,w)\in E$ with   paths $(v,u_{a,b}^{v,w},w^v)$ for all $(a,b)\in N_{v,w}$ (see Figure~\ref{fgraph}). The edges of the form 
$(u_{a,b}^{v,w},w^v)$ 
(the dashed edges)
are called \emph{dummy edges} and the edges of the form $(v,u_{a,b}^{v,w})$ 
(the solid edges)
are called \emph{label edges}. We  say
 that label edge $(v,u_{a,b}^{v,w})$  assigns label $a$ to $v$ and label $b$ to $w$.
 We will denote the obtained component by $G_{v,w}$ and we will use $E^l_{v,w}$ to denote the set of all label edges of $G_{v,w}$, obviously $|E^l_{v,w}|=|N_{v,w}|$.
 We finish the construction of $G'$ by 
 adding additional vertex~$s$ and  connecting all the components by additional dummy edges $(s,v)$ for all $v\in V$. A sample graph $G'$,
 where $G$ is $K_{3,3}$,
  is shown in Figure~\ref{fig2}. 
 \begin{figure}
    \begin{center}
      \psfrag{v}{\footnotesize $v$} 
      \psfrag{w}{\footnotesize $w$}
      \psfrag{wv}{\footnotesize $w^v$}
      \psfrag{Gvw}{$G_{v,w}$}
      \psfrag{u1}{\footnotesize $u^{v,w}_{a_1,b_1}$}
      \psfrag{u2}{\footnotesize $u^{v,w}_{a_2,b_2}$}
      \psfrag{un}{\footnotesize $u^{v,w}_{a_{|N_{v,w}|},b_{|N_{v,w}|}}$}
      \includegraphics{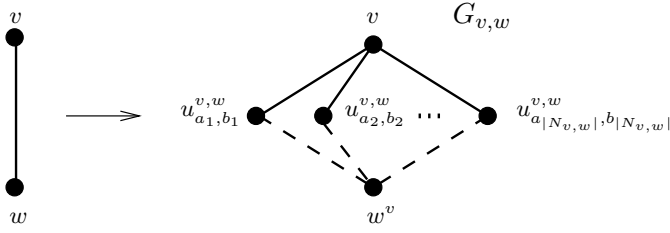}
    \end{center}
    \caption{Replacing edge $(v,w)\in E$ 
    with component~$G_{v,w}$.} \label{fgraph}
  \end{figure}
\begin{figure}
    \begin{center}
      \psfrag{s}{\footnotesize $s$} 
      
      \psfrag{v1}{\footnotesize $v_1$}
      
      \psfrag{w1v1}{\footnotesize $w_1^{v_1}$}
      \psfrag{Gv1w1}{\footnotesize $G_{v_1,w_1}$}
      
      \psfrag{w2v1}{\footnotesize $w_2^{v_1}$}
      \psfrag{Gv1w2}{\footnotesize $G_{v_1,w_2}$}

      \psfrag{w3v1}{\footnotesize $w_3^{v_1}$}
      \psfrag{Gv1w3}{\footnotesize $G_{v_1,w_3}$}
      
      \psfrag{v2}{\footnotesize $v_2$}
      
      \psfrag{w1v2}{\footnotesize $w_1^{v_2}$}
      \psfrag{Gv2w1}{\footnotesize $G_{v_2,w_1}$}
      
      \psfrag{w2v2}{\footnotesize $w_2^{v_2}$}
      \psfrag{Gv2w2}{\footnotesize $G_{v_2,w_2}$}

      \psfrag{w3v2}{\footnotesize $w_3^{v_2}$}
      \psfrag{Gv2w3}{\footnotesize $G_{v_2,w_3}$}
      
      \psfrag{v3}{\footnotesize $v_3$}
      
      \psfrag{w1v3}{\footnotesize $w_1^{v_3}$}
      \psfrag{Gv3w1}{\footnotesize $G_{v_3,w_1}$}
      
      \psfrag{w2v3}{\footnotesize $w_2^{v_3}$}
      \psfrag{Gv3w2}{\footnotesize $G_{v_3,w_2}$}

      \psfrag{w3v3}{\footnotesize $w_3^{v_3}$}
      \psfrag{Gv3w3}{\footnotesize $G_{v_3,w_3}$}

      \includegraphics{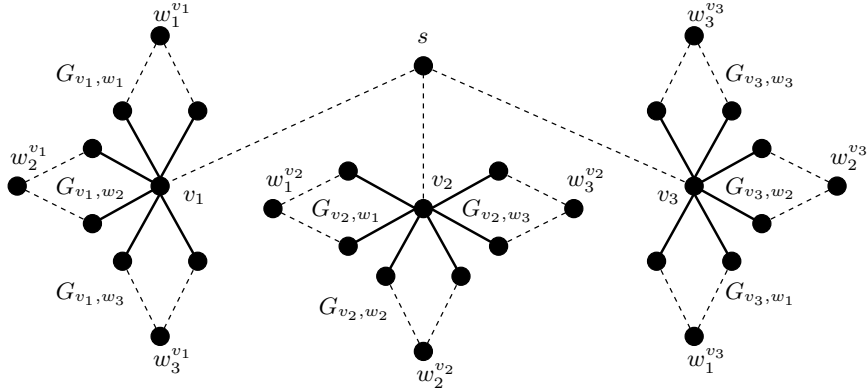}
    \end{center}
    \caption{A sample of graph $G'$, where graph~$G$ in $\mathcal{L}$ is $K_{3,3}$.} \label{fig2}
  \end{figure}

We now form scenario set $\Gamma$. We first note that all dummy edges under all scenarios  have costs equal to 0. We say that two label edges are \emph{label-distinct} if they do not assign the same label to any vertex $v$ or $w$. Namely, $(v,u_{a_i,b_i}^{v,w})$ and $(v',u_{a'_i,b'_i}^{v',w'})$ are label-distinct if $a_i=a_i'$ implies $v\neq v'$ and $b_i=b_i'$ implies $w\neq w'$.
Consider vertex $v\in V$, for which there is the set of $p=|\delta(v)|$ components  $\mathcal{G}=\{G_{v,w_1},\dots,G_{v,w_p}\}$. For every subset $\mathcal{F}\subseteq \mathcal{G}$ of exactly $g$ components,
$\mathcal{F}=\{G_{v,w_1},\dots,G_{v,w_g}\}$ and for every $g$-tuple of pairwise label-distinct edges 
$((v,u_{a_1,b_1}^{v,w_1}),\dots,(v,u_{a_g,b_g}^{v,w_g}))\in E^l_{v,w_1}\times\dots\times E^l_{v,w_g}$ we form scenario under which all these edges have cost~1 and all the remaining edges have cost~0. We repeat this procedure for all vertices $v\in V$. Consider then vertex $w\in W$, for which there is the set of $q=|\delta(w)|$ components $\mathcal{G}=\{G_{v_1,w},\dots,G_{v_q,w}\}$. For every subset $\mathcal{F}\subseteq \mathcal{G}$ of exactly $g$ components,
$\mathcal{F}=\{G_{v_1,w},\dots,G_{v_g,w}\}$ and for every $g$-tuple of pairwise label-distinct edges
 $((v_1,u_{a_1,b_1}^{v_1,w}),\dots,(v_g,u_{a_g,b_g}^{v_g,w}))\in E^l_{v_1,w}\times\dots\times E^l_{v_g,w}$ we form scenario under which all these edges have cost~1 and all the remaining edges have cost~0. We repeat this for all vertices $w\in W$.
 In order to ensure $\Gamma\not=\emptyset$,
we  include in $\Gamma$ the scenario
in which every edge has zero cost.

Assume that $\mathbf{w}\in L$ and thus $val(\mathcal{L})=1$.
Thus, there exists a total labeling $l$ satisfying all edges in $G$
such that  $\max_{x\in V\cup W}|l(x)|=1$. 
Each edge $(v_i,w_i)\in E$ in $G$ corresponds to the exactly one
  component~$G_{v_i,w_i}$
 in $G^{'}$. 
 Let $(a_i,b_i)$ be the pair of labels satisfying the edge~$(v_i,w_i)$ in 
 total labeling~$l$, i.e. $a_i\in l(v_i)$ and $b_i\in l(w_i)$.
 We form a spanning tree $T$ in $G'$ by adding exactly one edge $(v_i,u^{v_i,w_i}_{a_i,b_i})$ from every component $G_{v_i,w_i}$ and we complete the construction by adding a necessary number of dummy edges. 
 Since the labeling $l$ is such that $\max_{x\in V\cup W}|l(x)|=1$, 
  no pair of label-distinct edges have been chosen  while constructing $T$, 
  so $\sum_{e\in T} c^{S}_e\leq 1$ for all $S\in \Gamma$ and consequently
   $\max_{S\in \Gamma} \sum_{e\in T} c^{S}_e \leq 1$.

Assume that $\mathbf{w}\notin L$ and thus $\max_{x\in V\cup W}|l(x)|\geq N^\gamma\geq g$
 for all total  labellings~$l$. Consider any spanning tree $T$ in $G'$. Without loss of generality, we can assume that $T$ contains exactly one label edge from every component $G_{v,w}$. The set of all label edges contained in $T$ corresponds to a total labeling $l$ of $\mathcal{L}$.  Since $|l(x)|\geq g$,
  for some vertex $x\in V \cup W$, we have to use at least $g$  distinct labels in the labeling $l$. Suppose that $x=v\in V$ and we use distinct labels $a_1,\dots,a_g$ for $v$. Then, $T$ contains pairwise  label-distinct edges $(v,u^{v,w_i}_{a_i,b_i})$, $i=1,\dots,g$, and $\sum_{e\in T}c^{S}_e =g$ under scenario $S$ that correspond to this $g$-tuple of edges.  The reasoning for $x=w$, $w\in W$ is the same. In consequence $\max_{S\in \Gamma} \sum_{e\in T}c^{S}_e = g$ and  $OPT_1(\mathcal{T})=g$.

Let us now examine the size of the resulting instance of the \textsc{Min-max Spanning Tree} problem.
 The size of the set of edges $E^{'}$ is at most $|V|+2|E|N^2$,
 the size of the set of vertices $V^{'}$ 
  is   at most $1+|V|+|E|N^2+|W||V|$ and
 the number of scenarios is
 at most $ 1+2|E|^g N^g N^g$.
 Hence, 
 and from $|E|=|\mathbf{w}|^{O(\log N)}$, we deduce that
 the size  of the constructed instance
 $(G^{'},\Gamma)$
 is $|\mathbf{w}|^{O(g\log N)} N^{O(g)}$, so it can be constructed in $O(|\mathbf{w}|^{O(g\log N)} N^{O(g)})$ time.
\end{proof}

From Theorem~\ref{taprminmax}, we obtain 
the following  result:

\begin{thm}
  The \textsc{Min-max  Spanning Tree} problem with
  nonnegative edge costs under all scenarios
   is not approximable within $O(\log^{1-\epsilon} n)$ for any $\epsilon>0$,
  where $n$ is the input size, unless NP 
  $\subseteq$ DTIME$(n^{\mathrm{poly} \log n})$.
  \label{cmmmst}
\end{thm}
\begin{proof}

	Let $\gamma$ be the constant from Theorem~\ref{taprminmax}. For any $\beta>0$ we fix $g=\log^\beta |\mathbf{w}|$ and $N=\log^{O(\beta)} |\mathbf{w}|$, so that inequality $g\leq N^\gamma$ is satisfied for the constant $\gamma$ (see Theorem~\ref{taprminmax}).
  The input size of 
  the resulting instance~$(G^{'},\Gamma)$ from Theorem~\ref{taprminmax} is
  $n=|\mathbf{w}|^{O(g\log N)} N^{O(g)}=
   |\mathbf{w}|^{O(\log^{\beta+\delta}|\mathbf{w}|)}$ 
   for some constant $\delta>0$, so it can be constructed in 
   $O(|\mathbf{w}|^{\mathrm{poly}\log |\mathbf{w}|})$
    time. Since $g=\log^\beta |\mathbf{w}|$  and  $n=2^{O(\log^{\beta+\delta+1}|\mathbf{w}|)}$, we get $g= O(\log^{\frac{\beta}{\beta+\delta+1}}n)$
    and the gap is
    $O(\log^{1-\epsilon}n)$ for any $\epsilon>0$. 
   
\end{proof}

\begin{cor}
  The \textsc{Min-max  Regret  Spanning Tree} problem is not approximable within $O(\log^{1-\epsilon} n)$ for any $\epsilon>0$,
  where $n$ is the input size, unless NP 
  $\subseteq$ DTIME$(n^{\mathrm{poly} \log n})$.
  \label{cmmemst}
\end{cor}
\begin{proof}
 The corollary follows easily if we assume that each component $G_{v,w}$ in the construction from Theorem~\ref{taprminmax} has at least 2 label edges or, equivalently, every edge in the instance of \textsc{Label Cover} has at least two pairs of labels. In this case, under every scenario $S\in \Gamma$, there is a spanning tree of 0~cost (recall that we never assign two~1's to the same component in $S$). Hence $OPT_1(\mathcal{T})=OPT_2(\mathcal{T})$ and the proof is completed. If some edge in the instance of \textsc{Label Cover} has only one pair of labels, then this pair trivially forces an assignment of labels to two vertices, which (after checking consistency with other edges) can be removed from the instance before applying the construction from Theorem~\ref{taprminmax}.

\end{proof}

Up to this point we have assumed that the edge costs under all scenarios are nonnegative.
The following 
theorem  demonstrates that violation of this assumption makes the \textsc{Min-max Spanning Tree} problem not at all approximable:

\begin{thm}
      If both positive and  
	negative costs are allowed, then the 
	\textsc{Min-max  Spanning Tree} problem
	is not at all approximable unless P=NP  even for edge series-parallel graphs 
	\label{tmmna}
\end{thm}
\begin{proof}
We show a gap-introducing reduction from \textsc{3-SAT}
which is known to be strongly NP-complete~\cite{GJ79}.
\begin{description}
\item[\mdseries \scshape 3-SAT:]
\emph{Input:}  A set $U=\{x_1,\dots,x_n\}$ of 
Boolean variables and a collection 
$C=\{C_1,\dots,C_m\}$ of clauses, 
where every clause in $C$  has exactly 
three distinct literals.

\noindent\emph{Question:} If there is an assignment 
to~$U$ that satisfies all clauses in~$C$?
\end{description}

We will
assume that in the instance of \textsc{3-SAT} for every variable $x_i$ both $x_i$ and $\sim x_i$ appear in~$C$. Obviously,
under such assumption~\textsc{3-SAT} remains strongly 
NP-complete.
Given an instance of~\textsc{3-SAT} we construct an instance of \textsc{Min-max Spanning Tree} as follows.  For each clause $C_i=(l_i^1\vee l_i^2 \vee l_i^3)$
we create a graph~$G_i$ composed of~5 vertices: 
$s_i,v_1^i,v_2^i,v_3^i,t_i$ and 6 edges: the edges
$(s_i,v_1^i)$, $(s_i,v_2^i)$, $(s_i,v_3^i)$
correspond to literals in $C_i$,
the edges $(v_1^i,t_i)$, $(v_2^i,t_i)$, $(v_3^i,t_i)$
have  costs equal to~$-1$  under every scenario.
In order to construct a connected graph~$G=(V,E)$
with $|V|=4m+1$, $|E|=6m$, we 
 identify vertex~$t_i$ of $G_i$ with vertex~$s_{i+1}$ of~$G_{i+1}$
 for $i=1,\ldots m-1$.
 Note that the resulting graph~$G$ is edge series-parallel.  
 Finally,
 we form scenario set $\Gamma$ as follows.
For every pair of edges of $G$, $(s_i,v_j^i)$ and $(s_q,v_r^q)$,
that correspond to contradictory
 literals  $l_i^j$ and $l_q^r$, i.e. $l_i^j=\sim l_q^r$, we create scenario $S$ such that under this scenario the costs of the 
 edges $(s_i,v_j^i)$ and $(s_q,v_r^q)$ are set to~$4m-1$ and the costs of all the remaining edges are set to~$-1$.  It is easy to
 verify that each spanning tree~$T$ in the constructed 
 instance  has
 nonnegative maximal cost  over all scenarios.
 
Suppose that \textsc{3-SAT} is satisfiable. Then there exists
a spanning tree~$T$ of $G$ containing exactly $4m$ edges that
 do not correspond to contradictory literals. Thus,
 under every scenario~$S$, the tree contains at most one edge with the cost $4m-1$ and all the remaining $4m-1$ edges have costs equal to~$-1$. In consequence we get
 $\sum_{e\in T}c^{S}_e=0$ under every $S\in \Gamma$ and $OPT_1=0$.
 If \textsc{3-SAT} is unsatisfiable, then
 every spanning trees~$T$  of $G$ contains at least two edges
 which correspond to  contradictory literals, and
 so $OPT_1=\max_{S\in \Gamma}\sum_{e\in T}c^{S}_e \geq 4m$.
 Consequently
 \textsc{Min-max  Spanning Tree}
	is not approximable, unless P=NP. Otherwise,
	any polynomial time approximation
	algorithm applied to the constructed instance
	could decide if an instance of \textsc{3-SAT} is satisfiable.
\end{proof}

\subsection{Randomized algorithm
for min-max  spanning tree} \label{randmax}

If the edge costs are nonnegative under all scenarios, then  the \textsc{Min-max Spanning Tree} problem 
is approximable within~$K$, $K$ is the number of scenarios, 
and this is
the best approximation ratio known so far~\cite{ABV06}. On the other hand the problem is not at all approximable if negative costs are allowed (Theorem~\ref{tmmna}).
In this section, we assume that all costs are nonnegative and we give a polynomial time approximation algorithm
for the problem which
returns an $O(\log^2 n)$-approximate spanning tree, where $n$ is the number of vertices of $G$.
The algorithm is based on a randomized rounding of  a solution
to an iterative linear program.

It is easy to check that binary solutions to the following program
 $LP_{\min\max}(C)$ are in one-to-one correspondence with
 solutions to \textsc{Min-max Spanning Tree} of
 edge costs in every scenario at most $C$:
\begin{eqnarray}
	LP_{\min\max}(C):
		&&\sum_{e\in E}c^{S}_ex_e\leq C\;\;\forall_{S\in \Gamma},\label{c2}\\
		&&\sum_{e\in E}x_e=n-1,\label{d1}\\
	      &&\sum_{e\in \delta(W)}x_e \geq 1
	      \;\;\forall_{W\subset V}, \label{d2}\\
				&&0\leq x_e\leq1 \;\;
		\forall_{e\in E}, \label{d3}\\
		&&\text{if }c^{S}_e>C\text{ then }x_e=0\;\;
		\forall_{e\in E}\text{ and } \forall_{S\in \Gamma},
\end{eqnarray}
where $\delta(W)$ denotes the cut determined by vertex set $W$, i.e.
$\delta(W)=\{ (i,j) \in E\, :\, i\in W, j\in V\setminus W\}$. 
The core of $LP_{\min\max}(C)$ 
(constraints~(\ref{d1})-(\ref{d3}))
is  the relaxation  of the \emph{cut-set formulation} for spanning 
tree~\cite{TLMLAW95}.
The polynomial time solvability of $LP_{\min\max}(C)$ follows from an efficient  polynomial time separation based on 
the min-cut problem (see~\cite{TLMLAW95}).
Solving $LP_{\min\max}(C)$ consists in rejecting all edges $e\in E$ having $c_e^S>C$ under some scenario $S\in \Gamma$ and solving then the resulting linear programming problem.
Using  binary search in $[0,(n-1)c_{\max}]$, where $c_{\max}=\max_{e\in E} \max_{S\in \Gamma} c_e^S$,
one can find the minimal 
value of parameter~$C$, for which there is
a feasible solution to $LP_{\min\max}(C)$. 
Let $\widehat{C}$ be
this minimal value
and let  $(\hat{x}_e)_{e\in E}$ be
a feasible solution to $LP_{\min\max}(\widehat{C})$. Clearly $\widehat{C}\leq OPT_1$. Furthermore,
 if $\hat{x}_e>0$, then $c^{S}_e\leq \widehat{C}$
 and thus $c^{S}_e\leq OPT_1$  for each scenario $S\in \Gamma$. 

We now give an algorithm that randomly rounds a feasible solution 
of  $LP_{\min\max}(\widehat{C})$ to 
an $O(\log^2 n)$-approximate min-max spanning tree
(see Algorithm~\ref{algminmax}).
\begin{algorithm}
Use  binary search in $[0,(n-1)c_{\max}]$ to find 
 the minimal value of $C$ such that there exists
 a feasible solution to $LP_{\min\max}(C)$, i.e., $\widehat{C}$ and
 $(\hat{x}_e)_{e\in E}$.\\
Initially $\hat{F}$ contains only vertices of $G$, that is $n$ components. \\
$r\leftarrow \lceil 2(11+\sqrt{21})\ln n\rceil$ \\
\For{$k\leftarrow 1$  \emph{\KwTo} r}
    {
    For all $e\in E$, add edge $e$ independently with probability~$\hat{x}_e$ to~$\hat{F}$. \\
    \If{$\hat{F}$ $\mathrm{is}$ $\mathrm{connected}$}{\KwExitFor}
}
 \If{$\hat{F}$ $\mathrm{is}$ $\mathrm{connected}$} {\Return{{\rm a spanning tree of} $\hat{F}$}}
  \caption{Randomized algorithm for \textsc{Min-max  Spanning Tree}
  \label{algminmax}
}
\end{algorithm}

Let us analyze Algorithm~\ref{algminmax}. Obviously the
algorithm is polynomial. The following lemma 
shows that the total cost of edges included in each iteration
under any scenario $S\in \Gamma$ is  $O(\ln n) OPT_1$
with probability at least $1-\frac{1}{n}$:
\begin{lem}
Let $\hat{E}_k$ be a set of edges added to $\hat{F}$ 
at iteration~$k$ of  Algorithm~\ref{algminmax}
and let $K\leq n^{\rho_2}$, $1\leq f\leq n^{\rho_3}$, where
$f$, $\rho_1$, $\rho_2$,  $\rho_3$ are nonnegative
constants such that $\rho_2+\rho_3\leq 3.92\cdot \rho_1$, 
$\rho_1\geq 2$.
Then
\begin{equation}
\label{bminmax}
\max_{S\in \Gamma}\sum_{e\in \hat{E}_k}c^{S}_e
\leq \left(\rho_1\ln n+1.5\right) 
 \left(1+ 2\sqrt{1+\frac{\ln K+ \ln f}{\rho_1\ln n}}\right)OPT_1
\end{equation}
holds with probability at least $1-\frac{1}{fn^{\rho_1-1}}$.
 \label{lminmax}
\end{lem}
\begin{proof}
See Appendix~\ref{dod}.
\end{proof}

We now analyze the feasibility of an output solution~$\hat{F}$.
Let $\hat{F}_k$ be the forest obtained from $\hat{F}_{k-1}$ after  the~$k$-th iteration.
Initially, $\hat{F}_0$, $\hat{F}_0\subset G$, has no edges.
Let $C_k$ denote the number of connected components  of~$\hat{F}_k$. Obviously,
$C_0=n$.
We say that an iteration~$k$ is ``successful'' if
either $C_{k-1}=1$ ($\hat{F}_{k-1}$ is connected) or
$C_{k}<0.9C_{k-1}$;
otherwise, it is ``failure''.
We now recall a result of Alon~\cite{A95} (see also~\cite{DRM05}).
His proof is repeated in Appendix~\ref{dod} for completeness.
\begin{lem}[Alon~\cite{A95}]
For every~$k$, the conditional probability that iteration~$k$ is  
``successful'', given any set of components in $\hat{F}_{k-1}$,
is at least $1/2$.
\label{lalon}
\end{lem}
 From Lemma~\ref{lalon}, it follows that the
probability of the event that 
iteration~$k$ is ``successful'' is at least $1/2$.
This is a lower bound on the probability of
 success of given any history.
Note that, if forest $\hat{F}_k$ is not
connected ($C_k>1$) then the number of ``successful''
iterations has been less than $\log_{0.9} n < 10\ln n$.
Let~$\mathrm{X}$ be a random variable denoting the number
of  ``successful''
iterations among $r$~performed iterations of the algorithm.
The probability $\mathrm{Pr}[\mathrm{X}<10\ln n]$ can be upper bounded
by $\mathrm{Pr}[\mathrm{Y}<10\ln n]$, where
$\mathrm{Y}=\sum_{k=1}^{r}\mathrm{Y}_k$ is the sum of~$r$ independent
 Bernoulli trials
such that $\mathrm{Pr}[\mathrm{Y}_k=1]=1/2$. This estimation
can be done, since we have 
a lower bound on  success of given any history.
Clearly, $\mathbf{E}[\mathrm{Y}]=r/2$.
We apply the Chernoff bound (see for instance~\cite{MR95})
and determine the values of~$\delta\in (0,1]$ and $r$
in order to fulfill the following inequality:
\begin{equation}
\mathrm{Pr}[\mathrm{X}<10\ln n]\leq 
\mathrm{Pr}[\mathrm{Y}<10\ln n]=
\mathrm{Pr}[\mathrm{Y}<(1-\delta)\mathbf{E}[\mathrm{Y}]]<
 \mathrm{e}^{-\mathbf{E}[\mathrm{Y}]\delta^2/2}=\frac{1}{n}.
 \label{rchbralg}
\end{equation}
It is easily seen that inequality~(\ref{rchbralg}) holds
if the following system of equations
\begin{equation}
\begin{cases}
 (1-\delta)r/2 =10\ln n,\\
 r\delta^2/4= \ln n
\end{cases}
\label{rsyst}
\end{equation}
holds true.
An easy computation for $\delta$ and $r$
in (\ref{rsyst}), shows that 
$r=2(11+\sqrt{21})\ln n,
\;\;\delta=\sqrt{\frac{2}{11+\sqrt{21}}}$.
Hence, after $r$ iterations,
$r= \lceil2(11+\sqrt{21})\ln n\rceil$, 
we obtain 
with probability at least $1-1/n$ a spanning tree.
By the union bound and Lemma~\ref{lminmax} (set $f=r$),
with probability at least $1-1/n$ 
in every iteration, $k=1,\ldots,r$,
the set of edges~$\hat{E}_k$
included at iteration~$k$ satisfies the bound~(\ref{bminmax}).
We conclude that after $r$ iterations, 
we get
with probability at least $1-2/n$ a spanning tree
whose total cost in every scenario  is $O(r\ln n)OPT_1$.
We have, thus proved the following theorem:
\begin{thm}
There is a polynomial time randomized algorithm for
 \textsc{Min-max  Spanning Tree} 
that returns with probability at least $1-\frac{2}{n}$  a solution 
whose total cost  in every scenario is
$O(\log^2 n) OPT_1$.
 \label{thm3}
\end{thm}

\section{2-stage spanning  tree}

In this section, we discuss 
the \textsc{2-stage spanning tree} problem
in robust optimization setting.
We show that the problem
 is
 hard to approximate   within a ratio of
$O(\log n)$
unless the problems in NP have
quasi-polynomial algorithms.
Then,we give   an LP-based randomized approximation
algorithm
with ratio of $O(\log^2 n)$.

\subsection{Hardness of approximation}

\begin{thm}
  The \textsc{2-Stage Spanning Tree} problem is not approximable within 
  any constant, unless P=NP, and 
  within $(1-\epsilon)\ln n$
  for any $\epsilon>0$, unless NP$\subseteq$DTIME$(n^{\log \log n})$.
  \label{t2smsp}
\end{thm}

\begin{proof}
We proceed with a cost preserving reduction from
\textsc{Set Cover} to 
\textsc{2-Stage Spanning Tree}. The reduction
is similar to that in~\cite{FFK06}  for
the 2-stage stochastic spanning tree.
\textsc{Set Cover} is defined as follows (see, e.g.,~\cite{AC95,GJ79}):
\begin{description}
\item[\mdseries \scshape Set Cover:]
\emph{Input:}
A ground set $\mathcal{U}=\{1,\dots,n\}$ and a collection of its subsets
$U_1,\ldots, U_m$ such that $\bigcup_{i=1}^{m} U_i=\mathcal{U}$.

A subcollection $I\subseteq \{1,\dots,m\}$ \emph{covers} $\mathcal{U}$ if 
$\bigcup_{i \in I} U_i=\mathcal{U}$, where $|I|$  is the
\emph{size of the subcollection}.

\noindent\emph{Output:}
A minimum sized subcollection that covers $\mathcal{U}$.
\end{description}

 The \textsc{Set Cover} problem
 is not approximable within 
  any constant, unless P=NP,  and 
  within $(1-\epsilon)\log n$
  for any $\epsilon>0$, unless 
  NP$\subseteq$DTIME$(n^{\log \log n})$, where
  $n$ is the size of the ground set 
   (see~\cite{BGS95, F98}).
For a given instance 
$\mathcal{C}=(\mathcal{U},U_1,\ldots,U_m)$ of~\textsc{Set Cover},
we construct an instance  
$\mathcal{T}=(G=(V,E), \Gamma)$ of
\textsc{2-Stage  Spanning Tree} as follows.
Graph~$G=(V,E)$ is a complete graph  
with $m+n+1$ vertices
$V=\{u_1,\ldots,u_m,1,\ldots,n,r\}$.
Vertices $u_1,\ldots,u_m$ correspond to $m$
subsets $U_1,\ldots,U_m$, vertices $1,\ldots,n$
correspond to $n$ elements of set~$\mathcal{U}$.
The costs of the edges~$(r,u_i)$, $i=1,\ldots,m$,
 in $G$ in the first stage are 
 set to~$1$ and  the costs of all the remaining edges in $G$ are set 
 to~$m+1$.
Now we form scenario set $\Gamma$ in the second
stage. Each scenario $S_{j}\in \Gamma$ corresponds 
to vertex $j$, $j=1,\ldots,n$. Let
$T_j=\{j\} \cup \{u_i\;:\; j\in U_i\}$ and let $(T_j,V\setminus T_j)$
be the cut separating $T_j$ from all other vertices of~$G$.
Each second stage scenario~$S_j$
 is defined as:
the costs of the edges from cut $(T_j,V\setminus T_j)$
 are 
 set to~$m+1$ and  the costs of the remaining edges in $G$ are set 
 to~$0$.
 
We now prove that there is a subcollection 
of size at most $k\leq m$ that covers~$\mathcal{U}$
if and only if there exists a spanning tree in $G$ of the maximum 2-stage cost at most~$k\leq m$.
Given 
a subcollection $U_{i_1},\ldots,U_{i_k}$
of size~$k$ that covers~$\mathcal{U}$.
In the first stage,
we include in $E_1$ the edges 
$(r,u_{i_j})$, where vertices $u_{i_j}$ correspond to
subsets  $U_{i_j}$, $j=1,\ldots,k$. The cost of~$E_1$ is equal to~$k$.
In the second stage, we augment $E_1$ to
form a spanning tree with edges of cost zero in
each scenario $S_{j}$, $j=1,\ldots,n$.
Hence, the maximum 2-stage cost of the obtained spanning tree
equals~$k$.
Conversely, let $T$ be
a spanning tree in $G$
with the maximum 2-stage cost at most~$k$. Hence,
this tree does not contain any edge with cost~$m+1$.
Consequently, in the first stage the tree
 contains $k'\leq k$ edges of the form $(r,u_{i_j})$,
$j=1,\ldots,k^{'}$, and in the second stage in 
each scenario it contains zero cost edges.
 The vertices~$u_{i_j}$
correspond to subsets $U_{i_j}$, $j=1,\ldots,k^{'}$.
It is easily seen that  any element $i\in \mathcal{U}$ must be covered   by at least one 
of subsets $U_{i_j}$, $j=1,\ldots,k^{'}$. Otherwise the
solution would contain an edge of cost~$m+1$.
Thus, $U_{i_j}$, $j=1,\ldots,k^{'}$, form a subcollection 
of the size at most~$k$ that covers~$\mathcal{U}$.

The presented reduction is cost preserving.
Hence, \textsc{2-Stage Spanning Tree}
has the same approximation bounds
as \textsc{Set Cover}.
\end{proof}

\subsection{Randomized algorithm
for 2-stage spanning tree}

In this section we construct a randomized approximation algorithm for \textsc{2-Stage Spanning Tree}, which is based on a similar idea as the corresponding algorithm for \textsc{Min-max Spanning Tree} (see Section~\ref{randmax}). Consider the following program $LP_{2stage}(C)$, whose binary solutions correspond to the solutions of \textsc{2-Stage Spanning Tree}: 

\begin{eqnarray*}
	LP_{2stage}(C):
		&&\sum_{e\in E}c_e x_e +
		\sum_{e\in E}c^{S}_e x^{S}_e
		\leq C\;\;\forall_{S\in \Gamma}\\
		&&\sum_{e\in E}(x_e+x^{S}_e)=n-1\;\forall_{S\in \Gamma}\\
	      &&\sum_{e\in \delta(W)}(x_e +x^{S}_e)\geq 1
	      \;\;\forall_{W\subset V}, \;\forall_{S\in \Gamma}\\
	      &&0\leq x_e, x^{S}_e\leq 1 \;\;
	      \forall_{e\in E},\; \forall_{S\in \Gamma}\\
	      &&\text{if }c_e>C\text{ then }x_e=0\;\;
		\forall_{e\in E}\\
		&&\text{if }c^{S}_e>C\text{ then }x^{S}_e=0\;\;
		\forall_{e\in E},\; \forall_{S\in \Gamma}
\end{eqnarray*}
The algorithm (Algorithm~\ref{alg2stage}) 
randomly rounds a feasible solution 
$\hat{x}_e$, $\hat{x}^{S}_e$, $S\in \Gamma$, $e\in E$,
of  $LP_{2stage}(\widehat{C})$, where
$\widehat{C}$ denotes
the minimal value of~$C$ 
for which there is a feasible solution to $LP_{2stage}(C)$.
\begin{algorithm}
$c_{\max}\leftarrow \max_{e\in E}\{c_e, 
\max_{S\in \Gamma} c_e^S\}$\\
Use 
binary search in $[0,(n-1)c_{\max}]$ to find 
 the minimal value of $C$ such that there exists
 a feasible solution of $LP_{2stage}(C)$, i.e., 
$\hat{x}_e$, $\hat{x}^{S}_e$, $S\in \Gamma$, $e\in E$.\\
Initially $\hat{F}^{S}$  contains 
only vertices of $G$ for $S\in \Gamma$.   \\
$r\leftarrow 
\lceil  (\sqrt{\ln n +\ln K}+\sqrt{21\ln n +\ln K} )^2 \rceil$ \\
\For{$k\leftarrow 1$  \emph{\KwTo} r}
    {
    \emph{In the first stage:} For all $e\in E$, choose edge $e$
     independently with probability~$\hat{x}_e$ 
     and add it to each $\hat{F}^{S}$ for $S\in \Gamma$. \\
    \emph{In the second stage:}  for every $S\in\Gamma$
    and every $e\in E$, add edge $e$
     independently with probability~$\hat{x}^{S}_e$ 
     to~$\hat{F}^{S}$. \\
}
\If{{\rm all} $\hat{F}^S, S\in \Gamma$, {\rm are connected}}{
\Return{$\{\hat{F}^{S}\}_{S\in \Gamma}$}}
  \caption{Randomized algorithm for \textsc{2-stage  Minimum Spanning Tree}
  \label{alg2stage}
}
\end{algorithm}

An analysis of Algorithm~\ref{alg2stage}
proceeds similarly as the one of Algorithm~\ref{algminmax}.
The following lemma holds (the proof goes in similar manner as the proof of Lemma~\ref{lminmax}):
\begin{lem}
Let $\hat{E}_k$ and $\hat{E}^{S}_k$ be the sets of edges 
in the first stage and  in the second stage 
for every $S\in \Gamma$, respectively,
added to $\hat{F}^{S}$
at  iteration~$k$ 
of  Algorithm~\ref{alg2stage}
and let $K\leq n^{\rho_2}$, $1\leq f\leq n^{\rho_3}$, where
$f$, $\rho_1$, $\rho_2$,  $\rho_3$ are nonnegative
constants such that $\rho_2+\rho_3\leq 3.92\cdot \rho_1$, 
$\rho_1\geq 2$.
Then
\begin{equation}
\label{b2stage}
\sum_{e\in \hat{E}_k}c_e
+\sum_{e\in \hat{E}^{S}_k}c^{S}_e
\leq  
\left(\rho_1\ln n+1.5\right) 
 \left(1+ 2\sqrt{1+\frac{\ln K+ \ln f}{\rho_1\ln n}}\right)OPT_3
 \;\; \forall_{S\in \Gamma}
\end{equation}
holds with probability at least $1-\frac{1}{fn^{\rho_1-1}}$.
 \label{l2stage}
\end{lem}
Let $\hat{F}^{S}_k$ be the forest for $S\in \Gamma$
after  the~$k$-th iteration of Algorithm~\ref{alg2stage},
Let $C^{S}_k$ denote the number of connected components 
 of~$\hat{F}^{S}_k$. 
  Again, 
we say that an iteration~$k$ is ``successful'' if 
either $C^{S}_{k-1}=1$  or
$C^{S}_{k}<0.9C^{S}_{k-1}$; 
otherwise it is ``failure''. 
The
probability of the event that 
iteration~$k$ is ``successful'' is at least $1/2$,
which is due to  Lemma~\ref{lalon}.

Consider any scenario~$S\in \Gamma$.
If forest $\hat{F}^{S}_k$ is not
connected then the number of ``successful''
iterations is less than $\log_{0.9} n < 10\ln n$.
We estimate $\mathrm{Pr}[\mathrm{X}<10\ln n]$
by $\mathrm{Pr}[\mathrm{Y}<10\ln n]$, where
 $\mathrm{X}$ is random variable denoting the number
of  ``successful''
iterations among $r$~iterations and
$\mathrm{Y}=\sum_{k=1}^{r}\mathrm{Y}_k$ is the sum of~$r$ independent
 Bernoulli trials
such that $\mathrm{Pr}[\mathrm{Y}_k=1]=1/2$,
$\mathbf{E}[\mathrm{Y}]=r/2$.
We use  the Chernoff bound 
and compute the values of~$\delta\in (0,1]$ and $r$
satisfying  the following inequality:
\begin{equation}
\mathrm{Pr}[\mathrm{X}<10\ln n]\leq 
\mathrm{Pr}[\mathrm{Y}<10\ln n]=
\mathrm{Pr}[\mathrm{Y}<(1-\delta)\mathbf{E}[\mathrm{Y}]]<
 \mathrm{e}^{-\mathbf{E}[\mathrm{Y}]\delta^2/2}=\frac{1}{nK}.
 \label{r2stageralg}
\end{equation}
This gives
$r= (\sqrt{\ln n +\ln K}+\sqrt{21\ln n +\ln K} )^2$
and $\delta=\frac{2\sqrt{\ln n +\ln K}}{\sqrt{\ln n +\ln K}+\sqrt{21\ln n +\ln K}}$.
Recall that $K$ is the number of scenarios. By the union bound,
the probability that a forest in at least one scenario~$S$ 
is not connected is less than $1/n$. Again,
by the union bound and Lemma~\ref{lminmax} (set $f=r$),
with probability at least $1-1/n$ 
in every $k$ iteration, $k=1,\ldots,r$,
the sets of edges~$\hat{E}_k$ and $\hat{E}_k^{S}$
for each $S\in \Gamma$,
included at iteration~$k$, satisfy the bound~(\ref{b2stage}).
 Thus,
after $r$ iterations, $r=\lceil (\sqrt{\ln n +\ln K}+\sqrt{21\ln n +\ln K} )^2\rceil$,
with probability at least $1-2/n$, we obtain  spanning
trees of cost 
$O(r\ln n)OPT_3$ in every scenario. 
We get the following theorem:
\begin{thm}
There is a polynomial time randomized algorithm for
 \textsc{2-stage  Minimum Spanning Tree} 
that returns with probability at least $1-\frac{2}{n}$  a 
spanning tree
whose  cost  in every scenario is
$O(\log^2 n)OPT_3$.
 \label{tra2s}
\end{thm}


\appendix
\section{Some proofs}
\label{dod}

\begin{proof}(Lemma~\ref{lminmax})
In order to prove the bound~(\ref{bminmax}),
we will apply a technique used in~\cite{KMU08,KZ09a}.
Consider any scenario $S\in \Gamma$.  Let us sort
the costs in $S$ in nonincreasing order $c^{S}_{e[1]}\geq
c^{S}_{e[2]}\geq\cdots\geq c^{S}_{e[m]}$, 
($m$ is the number of edges of $G$).
We partition the ordered set of edges~$E$ into groups  as follows.
The first group $G^{(1)}$ consists of edges $e[1],\ldots,e[j^{(1)}]$,
where $j^{(1)}$ is the maximum such that 
$\hat{x}_{e[1]}+\cdots+\hat{x}_{e[j^{(1)}]}\leq \rho_1\ln n$.
The subsequent groups $G^{(l)}$, $l=2,\ldots,t$,
are defined in the same way, that is
$G^{(l)}$ consists of edges $e[j^{(l-1)}+1],\ldots,e[j^{(l)}]$,
where $j^{(l)}$ is the maximum such that 
$ \hat{x}_{e[j^{(l-1)}+1]}+\cdots+\hat{x}_{e[j^{(l)}]}\leq \rho_1\ln n$.
The optimal value $OPT_1$ satisfies:
\begin{equation}
OPT_1\geq  \widehat{C}\geq
\sum_{i=1}^{m}c^{S}_{e[i]}\hat{x}_{e[i]}
\geq
\sum_{l=1}^{t}\left[(\min_{e\in G^{(l)}}c^{S}_e)\sum_{e\in G^{(l)}}\hat{x}_e\right]
\geq 
(\rho_1\ln n-1)\sum_{l=1}^{t-1}\min_{e\in G^{(l)}}c^{S}_e.
\label{b10}
\end{equation}
Let $\mathrm{X}_e$ be a binary random variable with
$\mathrm{Pr}[\mathrm{X}_e=1]=\hat{x}_e$. It holds 
\begin{eqnarray}
\sum_{e\in \hat{E}_k}c^{S}_e
&\leq& 
\sum_{l=1}^{t}\sum_{e\in G^{(l)}}c^{S}_e \mathrm{X}_e
\leq
\sum_{l=1}^{t}\sum_{e\in G^{(l)}} (\max_{e\in G^{(l)}} c^{S}_e) \mathrm{X}_e
\nonumber\\
&\leq&
(\max_{e\in G^{(1)}}c^{S}_e) \sum_{e\in G^{(1)}}\mathrm{X}_e+
 \sum_{l=2}^{t}\left[(\min_{e\in G^{(l-1)}}c^{S}_e)\sum_{e\in G^{(l)}} \mathrm{X}_e\right].
 \label{b2}
\end{eqnarray}
Let us recall a Chernoff bound (see e.g.,~\cite{MR95}).
Suppose $\mathrm{X}_1,\ldots,\mathrm{X}_N$ are independent
Poisson trials such that $\mathrm{Pr}[\mathrm{X}_i=1]=p_i$.
Let $\mathrm{X}=\sum_{i=1}^{N}\mathrm{X}_i$
Then the inequality holds:
$\mathrm{Pr}[\mathrm{X}>\mathbf{E}[\mathrm{X}](1+\delta)]<
     \mathrm{e}^{-\mathbf{E}[\mathrm{X}]\delta^2/4}$ for any 
     $\delta\leq 2\mathrm{e}-1$.
We  use this Chernoff bound to estimate 
$\sum_{e\in G^{(l)}} \mathrm{X}_e$ in each group $G^{(l)}$.
Consider a group $G^{(l)}$. It holds $\mathbf{E}[\sum_{e\in G^{(l)}} \mathrm{X}_e]=\sum_{e\in G^{(l)}} \hat{x}_e
\leq \rho_1\ln n$. Set 
$\delta=2\sqrt{(\rho_1\ln n+\ln K+\ln f)/(\rho_1\ln n)}$.
Since $K\leq n^{\rho_2}$, $1\leq f\leq n^{\rho_3}$ and
$\rho_2+\rho_3\leq 3.92\cdot \rho_1$, $\rho_1\geq 2$,
inequality $\delta\leq 2\mathrm{e}-1$ holds.
Thus the Chernoff
bound yields:
\begin{equation}
\label{prob}
\mathrm{Pr}\left[\sum_{e\in G^{(l)}} \mathrm{X}_e> 
\rho_1\ln n(1+\delta)\right]
<\mathrm{e}^{-(\rho_1\ln n+\ln K+\ln f)}=1/(fKn^{\rho_1}).
\end{equation}
By the union bound, the probability that $\sum_{e\in G^{(l)}} \mathrm{X}_e> 
\rho_1\ln n(1+\delta)$ holds for at least one group $G^{(l)}$ is less than $1/{(fKn^{\rho_1-1})}$ (because the number of groups is at most~$n$). Now applying the bound $\sum_{e\in G^{(l)}} \mathrm{X}_e\leq
\rho_1\ln n(1+\delta)$ for every $l=1,\dots,t$ to~(\ref{b2}) and 
using the fact that $\max_{e\in G^{(1)}}w^{S}_e\leq OPT_1$
and inequality~(\ref{b10}) we obtain:
$$
 \sum_{e\in \hat{E}_k}c^{S}_e\leq 
 \rho_1\ln n\left(1+2\sqrt{\frac{\rho_1\ln n+\ln K+\ln f}{\rho_1\ln n}}\right) \left(OPT_1 + \frac{OPT_1}{\rho_1\ln n -1}\right).
$$
An easy computation shows that:
		$\sum_{e\in \hat{E}_k}c^{S}_e\leq \left(\rho_1\ln n+1.5\right) 
 \left(1+ 2\sqrt{1+\frac{\ln K+\ln f}{\rho_1\ln n}}\right)OPT_1$.
The probability that the bound  fails for a given scenario $S$ is less than $1/(fKn^{\rho_1-1})$ so, by the union bound, the probability that it fails for at least one scenario $S\in \Gamma$ is less 
than~$1/(fn^{\rho_1-1})$.
\end{proof}

\begin{proof}(Lemma~\ref{lalon})
If $\hat{F}_{k-1}$ is connected then we are done. Otherwise,
let us denote by $H=(V_H,E_H)$ the graph obtained 
from $\hat{F}_{k-1}$ by contracting  its every connected components
to a single vertex.  An edge~$e$ is not included in $\hat{F}_{k}$
with probability~$1-\hat{x}_e$. Hence, the
probability that
any vertex~$v$ of $H$
remains   isolated is
\[
\prod_{e\in \delta(v)}(1-\hat{x}_e)\leq 
\mathrm{exp}(-\sum_{e\in \delta(v)}(1-\hat{x}_e))\leq 1/\mathrm{e},
\]
where $\delta(v)$ denotes the set of edges incident to~$v$.
The last inequality follows from the fact that
$\sum_{e\in \delta(v)}(1-\hat{x}_e)\geq 1$. By linearity of
expectation, the expected number of isolated vertices of $H$
is $|V_H|/\mathrm{e}$, and thus with the probability at 
least $1/2$ the number of isolated  vertices is at most 
$2|V_H|/\mathrm{e}$. Hence, the number of
connected components of $\hat{F}_k$ is at 
most
\[
\frac{2|V_H|}{\mathrm{e}}+\frac{1}{2}
\left(|V_H|- \frac{2|V_H|}{\mathrm{e}} \right)
=\left(\frac{1}{2}+\frac{1}{\mathrm{e}} \right)|V_H|<0.9|V_H|.
\]
Since $|V_H|=C_{k-1}$, the lemma follows.
\end{proof}

\end{document}